\documentclass[10pt,onecolumn,journal,draftcls,oneside]{IEEEtran}

\ifCLASSINFOpdf
   \usepackage[pdftex]{graphicx}
\else
   \usepackage[dvips]{graphicx}
\fi

\usepackage{url}
\usepackage{subfigure}
\usepackage{relsize}
\usepackage{amsmath}
\usepackage{amssymb}
\usepackage{amsthm}
\usepackage{amsfonts}
\usepackage{paralist}
\usepackage{multirow}
\usepackage{bm}
\usepackage{algorithm}
\usepackage{algorithmic}
\usepackage{authblk}
\usepackage{xcolor}
\usepackage{cancel}

\newtheorem{theorem}{Theorem}

\newcounter{cond}
\newcounter{rema}
\newtheorem{remark}[rema]{Remark}

\newcounter{exam}
\DeclareFontFamily{OT1}{pzc}{}
\DeclareFontShape{OT1}{pzc}{m}{it}%
              {<-> s * [1.3] pzcmi7t}{}
\DeclareMathAlphabet{\mathpzc}{OT1}{pzc}%
                                 {m}{it}

\newcommand{\CN}{\mathcal{N}}
\newcommand{\CH}{\mathcal{H}}
\newcommand{\Cc}{\mathpzc{c}}

\newcommand{\CG}{\mathcal{G}}

\newcommand{\Cs}{\mathpzc{s}}

\newcommand{\Bd}{\boldsymbol{d}}

\newcommand{\Br}{\boldsymbol{r}}

\newcommand{\Bt}{\boldsymbol{t}}

\newcommand{\Bg}{\boldsymbol{g}}

\newcommand{\Bpi}{\boldsymbol{\pi}}

\begin{document}




\title{{\huge Polynomial Complexity Minimum-Time Scheduling \\ in a Class of Wireless Networks}}


\author[]{Qing He}
\author[1]{Vangelis Angelakis}
\author[1,2]{Anthony Ephremides}
\author[1]{Di Yuan}
\affil[1]{Department of Science and Technology, Link{\"o}ping University, Sweden}
\affil[2]{Department of Electrical and Computer Engineering,
University of Maryland, USA}
\maketitle

\begin{abstract}

We consider a wireless network with a set of transmitter-receiver pairs, or links, that share a common channel, and address the problem of emptying finite traffic volume from the transmitters in minimum time. This, so called, minimum-time scheduling problem has been proved to be ${\cal NP}$-hard in general. In this paper, we study a class of minimum-time scheduling problems in which the link rates have a particular structure consistent with the assumed environment and topology. We show that global optimality can be reached in polynomial time and derive optimality conditions. Then we consider a more general case in which we apply the same approach and thus obtain approximation as well as lower and upper bounds to the optimal solution. Simulation results confirm and validate our approach.



\end{abstract}

\let\thefootnote\relax\footnote 
{This work is supported in part by the Excellence Center at Link{\"o}ping-Lund in Information Technology, in part by the Swedish Research Council, in part by the EC Marie Curie Actions Projects MESH-WISE under Grant FP7-PEOPLE-2012-IAPP:324515 and Career LTE under Grant FP7-PEOPLE-2013-IOF:329313, and in part by the National Science Foundation under Grant CCF-0728966.

Q. He, V. Angelakis and D.Yuan are with the Department of Science and Technology, Link{\"o}ping University, Norrk{\"o}ping SE-60174, Sweden. E-mail: qing.he@liu.se; vanan@itn.liu.se; diyua@itn.liu.se. Tel: (+46)11363622; (+46)11363005; (+46)11363192.

A. Ephremides is with the Department of Science and Technology, Link{\"o}ping University, Norrk{\"o}ping SE-60174, Sweden, and also with the Department of Electrical and Computer Engineering, University of Maryland, College Park, MD 20742 USA. E-mail: etony@umd.edu. Tel: (+1)301405364.}

\section{Introduction}
\label{sec:introduction}

In a wireless network with a shared multiple access channel, the minimum-time scheduling problem amounts to determining which links are allowed to transmit simultaneously, and for how long, so that a given finite traffic volume at the sources can be delivered in minimum time. We consider the case in which all the links can be divided into several clusters, and in each cluster the activated links transmit at the same rate, which is determined by the numbers of simultaneously activated links in respective clusters. We illustrate in Figure \ref{fig:rn} an example scenario that corresponds to this case. Suppose a base station is located at a central point and the users are distributed along co-centric circles around it. Then users on the same circle have equal distance to the base station and hence identical geometric channel gain if the propagation environment is isotropic.

\begin{figure} [ht!]	
\centering
{\includegraphics[width=99mm]{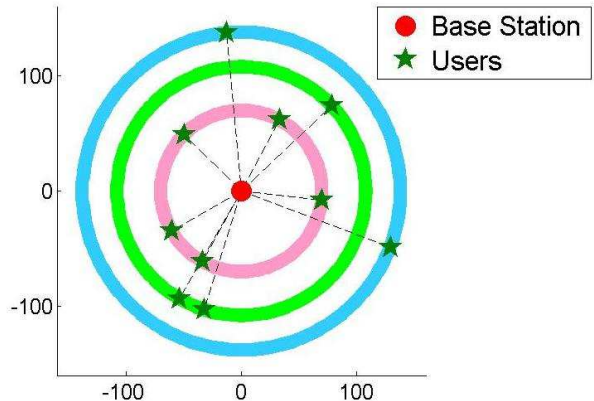}}
\vspace{-2mm}
\caption{An example wireless network}
\label{fig:rn}
\end{figure}

Assume all the users on the same circle transmit with the same power, we refer to the links with users on the same circle as a cluster. The links of the same cluster transmit at the same rate but that rate does depend on how many of them transmit at the same time. The reason is that, the total interference a link experiences only depends on how many links in each cluster are enabled instead of the individual elements of the activated set. It is worth noting that the scenarios corresponding to this case are not restricted to this simple topology. In fact, any case that has such multi-cluster symmetry in link rates is included.

\subsection{Related Work}
The Minimum-Time Scheduling Problem (MTSP), also known as the Minimum-Length Scheduling Problem, represents a fundamental aspect of resource allocation in wireless networks. The scheduling problem has been studied either through classical, so called protocol models \cite{HaSa88,Pr89,StAm90} or through so called physical model (e.g., \cite{BjVaYu03,GrHaNi00}). The latter enables successful decoding in the presence of interference and uses a cross-layer view of transmission rate and access control, thus making it possible to integrate this with general network resource allocation problems (e.g., \cite{BoEp06,CaFiGuYu13,GeNeTa06}).

In solving MTSP with signal-to-interference-and-noise ratio (SINR) constraints, several algorithms have been proposed. Notable amongst them is a column-generation-based solution method that was used in \cite{KoWi10}, which can approach or attain an optimal solution, with the advantage of a potentially reduced complexity. In \cite{PaEp08} it was formulated as a shortest path problem on directed acyclic graphs and the authors obtained suboptimal analytic characterizations. In \cite{isita2}, a modular algorithmic framework was provided to encompass exact as well as sub-optimal, but fast, scheduling algorithms, all under a unified principle design.

Previous complexity analysis \cite{Ar84,BjVaYu04,BoLiXi10,GoPsWa07} concluded the hardness of the scheduling problem with discrete rates, that is, a discrete set corresponding to SINR threshold. In \cite{j12}, 
we proved that the problem remains ${\cal NP}$-hard for most general cases of continuous rate functions. These ealier fundamental results provide the motivation for investigating the case that is on one hand polynomially tractable and on the other hand has high potential to approximate general practical scenarios.

\subsection{Contributions and Organization}
In this paper, our contributions include the following. First, we consider the class of MTSP with multi-cluster cardinality-based rates, proving it can be solved in polynomial time. Then we derive the optimality conditions for the solutions that correspond to emptying each cluster separately. Finally we apply this model to more general scenarios with k-means clustering and propose a column generation algorithm that can either precisely determine the optimal schedule, or provide a good approximation, as well as upper and lower bounds, to the optimum.

The paper is organized as follows. In Section \ref{sec:system}, we discuss the system model of MTSP and introduce the concept of multi-cluster cardinality-based rates, followed by the complexity analysis in Section \ref{sec:complex}. Then the problem decomposition is studied in Section \ref{sec:oc}. In Section \ref{sec:algorithm}, we extend the application of this setup with proposed approach, with the numerical results presented in Section \ref{sec:simulation}, and in Section \ref{sec:conclusion}, we provide final concluding remarks.

\section{System Model}
\label{sec:system}
\subsection{Minimum-Time Scheduling Problem in Wireless Networks}
We consider a wireless network of $N$ transmitter-receiver pairs, or links. Each link $i$ has a finite amount of backlogged data bits $d_i$ at its transmitter's queue. Without loss of generality, assume that the entries in the demand vector $\Bd = (d_1,d_2, ... d_N)$ are in descending order. Let $\CH$ denote the union of all subsets of $\CN$, excluding the empty set. Clearly, $|\CH| = 2^N-1$. We use the term {\bf\emph{group}} to refer to a member $\Cc \in \CH$. Scheduling a group $\Cc$ of the $N$ links means that all links in $\Cc$ are concurrently activated, with a given fixed power, for some positive amount of time $t_\Cc$. For any group $\Cc$, the effective transmit rate $r_{i\Cc}$ of any link $i \in \Cc$ directly depends on the composition of the activation group. 

In all systems with meaningful physical interpretations, the service rate of any link in a group does not increase if the group is augmented, i.e., for any two groups $\Cc \subset \Cc'$ and $i \in \Cc \cap \Cc'$, $r_{i\Cc} \geq r_{i\Cc'}$. We refer to this as the rate monotonicity property.

The MTSP amounts to determining which groups should be selected and their activating durations, so that the queues are emptied in minimum time.
If all the link rates for all $2^N-1$ groups $\Cc \in \CH$ are known, the MTSP accepts the following linear programming (LP) formulation:
\begin{subequations}
\label{eq:lp}
\begin{align}
\min~~ & \sum_{\Cc \in \CH} t_\Cc \label{eq:lpobj}\\
\text{s.~t.}~~ & \sum_{\Cc \in \CH} r_{i\Cc} t_\Cc = d_i ~~i=1, \dots, N \label{eq:lpcons} \\
& \Bt \geq 0
\end{align}
\end{subequations}

Herein $\Bt$ represents the scheduling vector composed of $t_\Cc, \Cc \in \CH$. We use $\Bt^*$ to denote an optimal scheduling solution. $T^* = ||\Bt^*||_1 = \sum_{\Cc \in \CH}t^*_{\Cc}$ stands for the optimal schedule length. Notation $\CH^*$ is reserved for a set of groups that correspond to an optimum solution, that is, $\CH^* = \{\Cc \in \CH: t^*_{\Cc} > 0\}$. 

In some of the analysis later on, we also utilize the LP dual of (\ref{eq:lp}). Letting $\pi_i$ denote the dual variables, the dual formulation is as follows.

\begin{subequations}
\label{eq:dual}
\begin{align}
\max~~ & \sum_{i \in \CN} d_i \pi_i, \label{eq:dualobj}\\
\text{s.~t.}~~ & \sum_{i \in \Cc} r_{i\Cc} \pi_i \leq 1~~\Cc \in \CH, \label{eq:dualcons} \\
& \Bpi \geq 0.
\end{align}
\end{subequations}

\subsection{Multi-Cluster Cardinality-Based Rates}

For the class of MTSP in which the links can be divided into $K$, where $K \leq N$, clusters, we use $n_j$ to denote the number of links in cluster $j, j = 1, 2, ..., K$, with $n_1 + n_2 + ... + n_K = N$.
For each group $\Cc \in \CH$, $\Cc_j$ refers to the set of activated links in cluster $j$. Clearly, $\Cc_j$ is a subset of $\Cc$ and $\Cc = \Cc_1 \bigcup \Cc_2 \bigcup...\bigcup \Cc_K$. 

The link rates in this problem class have the following properties:

\begin{itemize}
\item The links in the same cluster have the same rate if they are activated.
\item The rate value of a link is fully determined by the number of activated links in every cluster and the cluster that the link belongs to, but not on the identity of the links.
\end{itemize}

That is, for any link $i \in \Cc_j \subseteq \Cc$, the rate $r_{i\Cc}$ is a function of $j$ and the cardinalities $|\Cc_1|, |\Cc_2|, ...,$ and $|\Cc_K|$. Define $\Bg(\Cc) = (|\Cc_1|, |\Cc_2|, ..., |\Cc_K|)$ as the profile of group $\Cc$ and let $\CG$ denote the union of $\Bg(\Cc)$ for all $\Cc \in \CH$. As for all $2^N -1$ groups, $|\Cc_j|$ ranges from $0$ to $n_j$, we have $|\CG| = \prod_{j=1}^K{(n_j+1)}-1$.
In this problem class, the rate of any link in cluster $j$ can be denoted by $r^j_{\Bg(\Cc)}$. Therefore we refer to it as Multi-Cluster Cardinality-based Rates (MCCR).

For MTSP with MCCR, we subsequently use $(N, \Bd, \Br, K)$ to refer to this problem class. The key notations used are summarized in Table \ref{tab:notation} for easy reading. 

Note that the problem is modelled in a rather generic form and the link rates are not restricted to be produced by any explicit or implicit rate functions. Thus the new insights presented in this paper are independent of physical-layer system specifications and are valid for any feasible, or achievable, rates from a communication/information-theoretic perspective.

\begin{table}
\caption{Notation}
\vspace{-2mm}
\label{tab:notation}
\centering
\begin{tabular}{ c | l }
\hline
\hline
 Notation  & ~~~~~~~~~ Description \\
\hline
\hline
$N$					& Number of links \\
$\CN$				& The set of $N$ links \\
$\CH$				& The union of subsets of $\CN$ \\
$d_i$				& Data bits in the queue of transmitter of link $i$ \\
$\Bd$				& Demand vector $(d_1,d_2, ... d_N)$ \\  
$\Cc$				& Group, member of $\CH$ \\
$r_{i\Cc}$			& Rate of link $i$ if group $\Cc$ is activated \\
$t_\Cc$				& Activation time of group $\Cc$ \\
$\Bt$				& Scheduling vector, composed of $t_\Cc, \Cc \in \CH$ \\
$\Bt^*$				& Optimal scheduling solution \\
$T^*$				& Optimal scheduling length \\
$\CH^*$				& The set of activated groups at optimal, that is, \\ 
					& $\CH^* = \{\Cc \in \CH: T^*_{\Cc} > 0\}$ \\
$K$					& Number of clusters \\
$n_j$				& Number of links in cluster $j$ \\
$\Cc_j$				& Link set in the intersection of group $\Cc$ and cluster $j$, $\Cc_j \subseteq \Cc$ \\
$\Bg(\Cc)$ 			& Profile of group $\Cc$, defined as $\Bg(\Cc) = (|\Cc_1|, |\Cc_2|, ..., |\Cc_K|)$\\ 
$\CG$ 				& The union of $\Bg(\Cc)$ for all $\Cc \in \CH$ \\
$r^j_{\Bg(\Cc)}$	& Rate of any activated link in cluster $j$ and group $\Cc$ \\ 
					& with profile $\Bg(\Cc)$
\\
\hline
\hline
\end{tabular}
\end{table}


\section{Complexity Consideration}
\label{sec:complex}

We show in the following that MTSP with MCCR can be solved in polynomial time as long as $K$ is independent of $N$.

\begin{theorem}
\label{theo:polynomialM}
$(N, \Bd, \Br, K)$ is in class P, that is, the global optimum can be computed in polynomial time.
\end{theorem}

\begin{proof}
Consider the LP dual problem in (\ref{eq:dual}). For problem class $(N, \Bd, \Br, K)$, the dual has the following form.
\begin{subequations}
\label{eq:dualc}
\begin{align}
\max~~ & \sum_{i \in \CN} d_i \pi_i \label{eq:dualcobj}\\
\text{s.~t.}~~ & \sum_{j=1}^K (r^j_{\Bg(\Cc)} \sum_{i \in \Cc_j \subseteq \Cc} \pi_i) \leq 1~~\Cc\in \CH \label{eq:dualccons} \\
& \Bpi \geq 0
\end{align}
\end{subequations}

Observe that in \eqref{eq:dualccons}, the dual variables of links in the same cluster have uniform coefficient, i.e., the same rate value, and occur in the same pattern in the constraint set. As a result, for any feasible solution, if we swap the values of $\pi_i$ and $\pi_l$, for any $i, l \in \Cc_j$, the new solution remains feasible. It follows that there exists an optimal solution such that, for each cluster, the order of values of dual variables is consistent with the order of values of link demands because otherwise the objective function value can be improved by swapping the variable values so that the condition holds. The demand vector $\Bd$ is given in descending order, with the result that in this optimal solution, $\pi_1 \geq \pi_2 \geq \dots \geq \pi_N  \geq 0$. Based on the above observation, one can conclude that, for all constraints in \eqref{eq:dualccons} that correspond to the groups with same profile $\Bg(\Cc) = (|\Cc_1|, |\Cc_2|, ..., |\Cc_K|)$, if the one that consists of the first $|\Cc_j|$ links in cluster $j$ is satisfied, the others are also met. Therefore, we can use this constraint, which is the most stringent one, as a substitute for the whole set.

For example, assume $K = 3$, in each cluster, arrange the link demands in descending order. For simplification, define the $i$th link in cluster $j$ as link $ji$. For all the groups with $\Bg(\Cc) = (|\Cc_1|, |\Cc_2|, |\Cc_3|) = (1, 2, 1)$, we can identify the most stringent constraint is the one constituted with the first link in cluster 1, the first two links in cluster 2 and the first link in cluster 3. That is, 

\begin{equation}
\label{eq:strict}
r^1_{\Bg(\Cc)}\pi_{11} + r^2_{\Bg(\Cc)}(\pi_{21} + \pi_{22}) + r^3_{\Bg(\Cc)}\pi_{31} \leq 1
\end{equation}


Following the aforementioned rationale, the inequality (\ref{eq:strict}) is equivalent to cover the constraints for all the
${n_1 \choose 1} {n_2 \choose 2} {n_3 \choose 1}$ groups, herein $n_j$ is the total number of links in cluster $j$.

For all $2^N-1$ possible groups, there are $|\CG|$, i.e., $\prod_{j=1}^K{(n_j+1)}-1$ different profiles $\Bg(\Cc)$. Following the above approach, in \eqref{eq:dualccons}, we have at most the same number of constrains that are sufficient to define the optimum. Therefore, together with the constraints derived from the order of link demands, i.e., $\pi_1 \geq \pi_2 \geq \dots \geq \pi_N \geq 0$, the total number of constraints reduces to $N-1+\prod_{j=1}^K{(n_j+1)}$, implying that the optimal solution to problem class $(N, \Bd, \Br, K)$ is found by solving an LP of size $O(N^K)$. As long as $K$ is independent of $N$, it can be solved in polynomial time, hence the conclusion follows.
\end{proof}

\begin{remark}
Theorem \ref{theo:polynomialM} significantly extends previous results on the complexity of cardinality-based rates model (see \cite{j12} and the references therein). In fact, the cardinality-based rates model is a special case of $(N, \Bd, \Br, K)$, in which $K=1$, that is, the case of only one cluster so that the rate values are determined solely by the group cardinality $|\Cc|$.
\end{remark}

\section{Optimality Conditions for Problem Decomposition}
\label{sec:oc}
According to the proof of Theorem \ref{theo:polynomialM}, the complexity of $(N, \Bd, \Br, K)$ is determined by $N^K$, which means if $K$ increases, the complexity grows rapidly. However, if the scheduling problem allows a decomposition into $K$ subproblems so the optimal solution for each cluster can be constructed separately, the complexity is significantly decreased. Thus we are interested in investigating when such a decomposition can occur. We refer to the groups formed by link(s) in a single cluster as {\bf\emph{intra-cluster groups}}, (as opposed to {\bf\emph{inter-cluster groups}}), which contain links in more than one cluster). Decomposition means that the optimal solution only needs to consider intra-cluster groups.

We first consider a more structured setting, in which the links in the same cluster have uniform demand. The following theorem provides a sufficient condition. 

\begin{theorem}
\label{theo:introclu}
$(N, \Bd, \Br, K)$ with uniform demand in each cluster decomposes if there is at least one intra-cluster group in each cluster having higher or equal sum-rate, i.e., the sum of all the link rates in the group, than any inter-cluster group.
\end{theorem}

\begin{proof}
Suppose $\CH^*$ is an optimum schedule.
We utilize the fact that the links of each cluster exhibit
full symmetry in demand and rate, and perform a two-step
transformation of $\CH^*$ to arrive at a new solution that uses
intra-cluster groups only and does not run longer than $\CH^*$ in
schedule length.

In the first step, we construct a new solution ${\bar
\CH}^*$ as follows. Consider any group $\Cc^* \in \CH^*$ that runs
for time $t_{\Cc^*}$.
Without loss of generality, for notational convenience we assume
$\Cc^*$ is composed by links of the first $k$ clusters, with $1 \leq k
\leq K$, i.e., $\Cc^* = \cup_{j=1}^k \Cc^*_j$.
Consider any cluster $j \in \{1, \dots, k\}$, and suppose,
again merely for simplifying notation, the link indices of this
cluster are $1, \dots, n_j$, among which $1,\dots, |\Cc^*_j|$ form
$\Cc_j^*$ in group $\Cc^*$. We construct $n_j$ groups, all having
size $|\Cc_j^*|$ and containing links of cluster $j$ in a rotational manner,
that is, $\{1, 2, \dots, |\Cc_j^*|\}$, $\{2,
\dots, |\Cc_j^*|, |\Cc_j^*|+1\}$, and so on, with the last group being $\{n_j, 1, 2,
\dots, |\Cc_j^*|-1\}$.
Applying this operation to all clusters $j = 1, \dots, k$ and taking
the Cartesian product of the resulting sets of groups lead to a set of
$\prod_{j=1}^k n_j$ link groups. Each group runs for a time duration $\frac{t_{\Cc^*}}{\prod_{j=1}^k n_j}$. 

Consider applying the above construction to all $\Cc^* \in \CH^*$, and
denote the result by ${\bar \CH}^*$. We make the following
observations. First, for any cluster, all of its links receive equal
activation time in ${\bar \CH}^*$.  This, together with the assumption
of uniform demand across all links in the same cluster, imply that for any link and group in ${\bar
\CH}^*$, the link will not have its demand emptied when the group is
active. Second, for each cluster having link activation in $\Cc^* \in
\CH^*$, the rate of any link of this cluster in the groups derived 
remain as that in $\Cc^*$, because
the cardinality of links of the cluster is not altered.
Third, again as a result of equal time sharing,
for any cluster $j$, the total amount of demand drained 
by ${\bar \CH}^*$ for all links in $j$ equals that in $\CH^*$.
In conclusion, ${\bar \CH}^*$ qualifies as a schedule.
By construction, ${\bar \CH}^*$ has the same length in
time duration as $\CH^*$, and hence ${\bar \CH}^*$
is an alternative optimum.

The second step of the transformation eliminates the use of inter-cluster
groups. Consider any group $\Cc^*$ in the original schedule $\CH^*$
and suppose again $\Cc^* = \cup_{j=1}^k \Cc^*_j$,
with $k>1$, i.e.,
$\Cc^*$ is an inter-cluster group. 
Denote by ${\bar \CH}^*(\Cc^*) =  \{ {\bar \Cc}^{*,1},
\dots, {\bar \Cc}^{*,m} \} $ the set of the groups derived from $\Cc^*$ in step one,
with $m = \prod_{j=1}^k n_j$. Consider any cluster $j \in \{1, \dots
k\}$. It follows from the construction above (cf. Fig. \ref{fig:proc}) that all
links of $j$ appears in ${\bar \CH}^*(\Cc^*)$, with full symmetry in
terms of both rate and amount of activation duration.  Thus the amount
of demand drained by using ${\bar \CH}^*(\Cc^*)$ is uniform for all
links in $j$. Moreover, activating any ${\bar \Cc}^{*,i} \in {\bar
\CH}^*(\Cc^*)$, the demand drained per time unit for all links (of
multiple clusters) in ${\bar \Cc}^{*,i}$ equals the sum rate of ${\bar \Cc}^{*,i}$.


By the theorem's assumption, there is some cardinality, say $\ell_j$,
for which the intra-cluster groups of cluster $j$ have the same or
better sum rate than that of any group in ${\bar \CH}^*(\Cc^*)$.  We
construct $n_j$ intra-cluster groups with cardinality $\ell \leq n_j$
for cluster $j$ in a rotational manner, that is, $\{1, \dots,
\ell_j\}$, $\dots$, $\{n_j, \dots, \ell_j-1\}$.
By the construction, using these $n_j$ intra-cluster groups with
uniform activation time drains the same amount of demand across
all links of $j$. Consider performing the
construction for all $j = 1, \dots, k$, and using the resulting 
intra-cluster groups to serve the same amount of demand as in ${\bar
\CH}^*(\Cc^*)$ for each $j$. Doing so achievers a higher or equal
amount of drained demand per time unit than ${\bar \CH}^*(\Cc^*)$
throughout, no matter which of the new groups is under activation.
Hence after the second step of the transformation we obtain a schedule 
consisting of intra-cluster groups only and having better or equal
performance as $\CH^*$, and the theorem follows.
\end{proof}

We illustrate the two-step transformation in Figure \ref {fig:proc} with an example.
Assume cluster 1 has link set $\{1, ~2, ~3, ~4\}$ and cluster 2 has $\{5, ~6,~7\}$. The intra-cluster groups in cluster 1 and 2, with two and three links, respectively, have higher sum-rates. In this case, we show below the transformation on an arbitrarily selected group $\Cc^* = \{1, ~2, ~3, ~5, ~6\}$.

\begin{figure} [ht!]	
\centering
{\includegraphics[width=119mm]{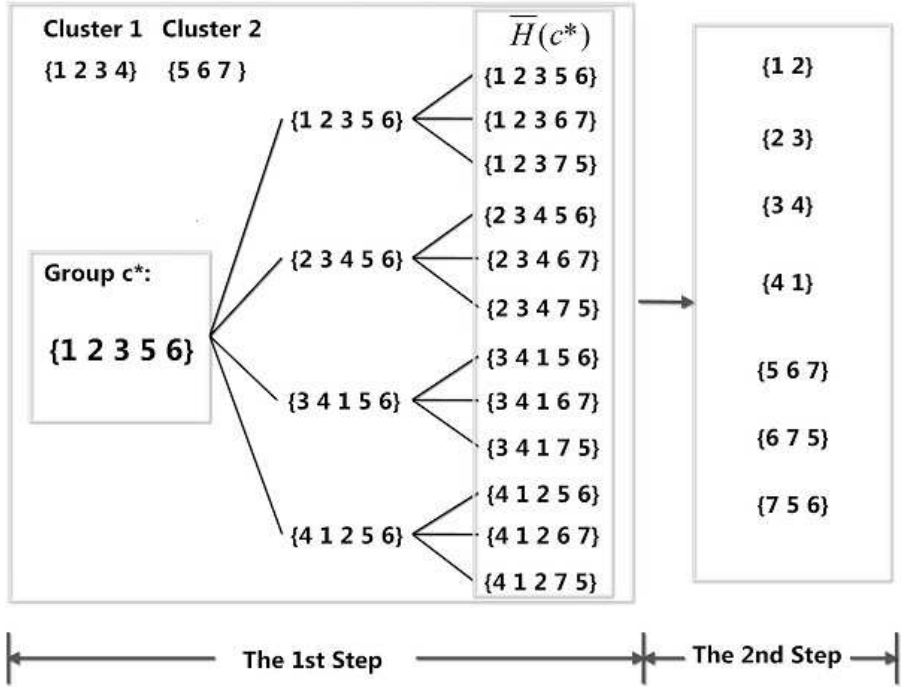}}
\vspace{-2mm}
\caption{Two-step transformation}
\label{fig:proc}
\end{figure}

For the links with non-uniform demand, the above theorem cannot be applied directly. Thus we set up the following sufficient condition for the general case.

\begin{theorem}
\label{theo:introclun}
If the single-link group in each cluster has higher or equal sum-rate than any inter-cluster group, then the problem $(N, \Bd, \Br, K)$ with non-uniform demand decomposes.
\end{theorem}  

\begin{proof}
Suppose an inter-cluster group $\Cc^* = \cup_{j=1}^k \Cc^*_j, 1 < k \leq K$, exists in an optimal solution and it runs for time $t_{\Cc^*}$. The solution can be further improved by replacing this group with a combination of single-link groups, that is, all the links in $\Cc^*$ will be activated one by one. To empty same demand as $\Cc^*$ does, the activation duration for the single-link groups of cluster $j$ is:
\begin{equation}
\frac {r^j_{\Bg(\Cc^*)}|c_j|t_{\Cc^*}} {R_j}, 
\end{equation}
herein $R_j$ is the rate of the link in cluster $j$ when it is activated alone.

According to the theorem's condition, the total time the new groups needed is calculated as following:
\begin{equation}
\sum_{j=1}^k \frac {r^j_{\Bg(\Cc^*)}|c_j|t_{\Cc^*}} {R_j} \leq \sum_{j=1}^k \frac {r^j_{\Bg(\Cc^*)}|c_j|t_{\Cc^*}} {\sum_{j=1}^k r^j_{\Bg(\Cc^*)}|c_j|} = t_{\Cc^*}
\end{equation}

The new solution is either better or as good as the original one, hence the conclusion follows.
\end{proof}

On recognizing the decomposition for $(N, \Bd, \Br, K)$, we show in the following theorem that it is polynomial-time tractable.

\begin{theorem}
\label{theo:ccc}
The computational complexity of recognizing the conditions in Theorem \ref{theo:introclu} and \ref {theo:introclun} is no more than $O((\frac{N+K}{K})^K)$.
\end{theorem}  

\begin{proof}
For $(N, \Bd, \Br, K)$, we have $2^N-1$ groups $\Cc \in \CH$. Among them, the groups with the same profile $\Bg(\Cc)$ have uniform sum-rate due to the properties of MCCR. Therefore, even in the worst case in which we need to check all the possible sum-rates to recognize the decomposition condition, there are no more than $|\CG| = \prod_{j=1}^K{(n_j+1)}-1$ sum-rates instead of $2^N-1$. According to the inequality of arithmetic and geometric means:

\begin{equation}
\prod_{j=1}^K{(n_j+1)} \leq \left(\begin{array}{c}\frac{\sum_{j=1}^K{(n_j+1)}}{K}\end{array}\right)^K = \left(\begin{array}{c}\frac{N+K}{K}\end{array}\right)^K 
\end{equation}
The computational complexity is at most $O((\frac{N+K}{K})^K)$. Hence the conclusion follows.
\end{proof}

\begin{remark}

Even with the maximum time complexity $O((\frac{N+K}{K})^K)$, the recognition of decomposition condition is much easier than directly solving the LP.
This is because the running time for the latter is of $O(N^{3.5K})$ (or even longer, depending on method).
Moreover, for some special cases, the recognition can be trivial. We show an example of such cases in Figure \ref{fig:exmp2}.
In this scenario, the receivers in one cluster are far from their own transmitters but close to the transmitters in the other cluster, thus the links in different clusters generate significant interference to each other, which means the sum-rate of an inter-cluster group is much less than that of an intra-cluster group. Thus we can disregard most of inter-cluster groups in checking the condition.
\end{remark}


\begin{figure} [ht!]	
\centering
{\includegraphics[scale=0.5]{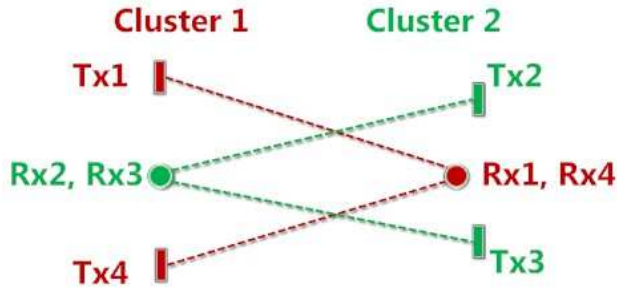}}
\vspace{-2mm}
\caption{An example with high interference among inter-cluster links.}
\label{fig:exmp2}
\end{figure}

\section{Application of Multi-Cluster Cardinality-Based Rates Model}
\label{sec:algorithm}

The above results hold provided that the links can be partitioned into $K$ clusters such that the symmetry in rates holds for all the links in the same cluster. In a wireless system with one base station and a number of users, this is not the case even with a distance-based propagation model because the users' locations do not confirm with the given criteria.
For this scenario, the complexity result in \cite{j12} is still valid, that is, with a polynomial reduction, the MTSP is proved to be as hard as a fractional coloring problem, which is known to be ${\cal NP}$-hard. Therefore the scheduling problem for this scenario remains hard. 
However, we can divide the users into several clusters in which they have similar distances to the base station. Then the MCCR can be applied to 
the model of the general case. The greater the total number of clusters $K$, the more accurate the modelling is. Therefore the applicability of MCCR extends to a more general setup as an approximation and its polynomial-time tractability contributes to the efficiency of the solution.

\subsection{k-means clustering}
To apply the MCCR model to general networks, the links need to be divided into several clusters with close-to-uniform gain within each cluster. For geometric distance based channel gain, an intuitive method is to partition the users into several clusters based on their locations so that each user belongs to the cluster with the nearest mean distance to the base station. This is actually the well-studied k-means clustering. That is, define $l_i$ as the length of link $i$, i.e., the distance between user $i$ and the base station, we aim to partition the $N$ data points in $\{l_1, l_2, ..., l_N\}$ into $K$ disjoint subsets $\Cs_j, j = 1, 2, ..., K$ so as to minimize the within-cluster sum of squares (WCSS):
\begin{subequations}
\label{eq:kmeans}
\begin{align}
\min~~ & \sum_{j=1}^K {\sum_{i \in \Cs_j}|l_i - \mu_j|^2},\\
\text{s.~t.}~~& \bigcup_{j=1}^K \Cs_j = \{l_1, l_2, ..., l_N\},\\
& \Cs_j \cap \Cs_m = \emptyset, \forall j \neq m,\\
& \mu_j = \frac {\sum_{i \in \Cs_j}l_i}{|\Cs_j|}.
\end{align}
\end{subequations}


Here we use the notation $K$ as a parameter in the clustering. Increasing $K$ on one side improves the precision of the modelling, but on the other side, increases the time complexity of solving the model. Hence a trade-off is needed. In Section \ref{sec:simulation}, numerical results are provided to show the impact of $K$ on performance.

We apply the standard k-means algorithm, which is also referred to as Lloyd's algorithm \cite{Md03}, to the problem. 
The algorithm consists of a re-estimation procedure as follows. Initially, the data points are assigned at random to the $K$ sets. In step 1, each data point $l_i$ is assigned to the cluster whose mean yields the least WCSS. In step 2, the new mean of the data points in each cluster is calculated to be the new centroid. These two steps are alternated until a stopping criterion is met, i.e., when there is no further change in the assignment of the data points. 

After clustering, we set the mean $\mu_j^*$ as the new length of the links in cluster $j$ for cardinality-based rates calculation. It is utilized to provide an approximation of actual link rate.

\subsection{Column Generation Algorithm}
The problem $(N, \Bd, \Br, K)$ is polynomial-time solvable. Still, when $N^K$ is large, directly solving the LP is time-consuming since its running time can be more than cubic of the LP size. Therefore we develop a column generation algorithm with the advantage of reduced computational complexity.

The algorithm works as follows: firstly we split the scheduling problem into two problems: the master problem and the subproblem. The master problem is the original LP (\ref{eq:lp}) with only a subset of variables, or groups (denoted by $\CH'$) being considered.

\begin{subequations}
\label{eq:cg}
\begin{align}
\min~~ & \sum_{\Cc \in \CH'} t_\Cc \label{eq:cgo}\\
\text{s.~t.}~~ & \sum_{\Cc \in \CH'} r_{i\Cc} t_\Cc = d_i ~~i=1, \dots, N \label{eq:cgcons} \\
& \Bt \geq 0
\end{align}
\end{subequations}

The subproblem is the following 
\begin{subequations}
\label{eq:rc}
\begin{align}
\min~~ &1- \sum_{i \in \CN} r_{i\Cc} \pi_i, \label{eq:rco}\\
\text{s.~t.}~~ & \sum_{i \in \Cc} r_{i\Cc} \pi_i \leq 1~~\Cc \in \CH, \label{eq:rcc} \\
& \Bpi \geq 0.
\end{align}
\end{subequations}

We start with a simple and feasible TDMA-based activation, that is, all the $N$ links are emptied one by one. Solving the master problem, 
we are able to obtain dual prices $\pi$ for each of the constraints. This information is then utilized in the objective function of the subproblem. On solving the subproblem, 
we design the following quick method based on the property of MCCR: 1) rank the $\pi_i$ with descending order in each cluster; 2) form the group that has the minimum reduced cost, i.e., the objective of subproblem, among all the candidate groups that have same $\Bg(\Cc)$. Thus we 
reduce the size of solution space of the subproblem from $2^N-1$ to $\prod_{j=1}^K{(n_j+1)}-1$. Then we solve the subproblem and if the objective value is negative, a variable, i.e., a group, with negative reduced cost has been identified. This variable is then added to the master problem, and the master problem is re-solved. Re-solving the master problem will generate a new set of dual values, and the process is repeated until no negative reduced cost variables are identified. The subproblem returns a solution with non-negative reduced cost. We can thus conclude that the solution to the master problem is optimal.

\subsection{Lower and Upper Bounds for Optimum}
Besides approximating the optimal solution by calculating link rates with $\mu^*$, the proposed approach can be utilized to obtain lower and upper bounds for MTSP. In detail, instead of the rates derived from $\mu^*$, we consider the minimum and maximum rates in each cluster. 

For the example scenario, the approximation is equivalent to put all the links in cluster $j$ on a circle around the central base station with a radius equals to $\mu_j^*$. Now we extend the circle to a minimum ring in which all the users in this cluster are included. Then if we use the minimum link length to calculate the signal strength and the maximum value to calculate the interference, we obtain an upper bound for the SINR. Conversely, we have a lower bound. As in wireless networks, the rate that a link can support is monotonously increasing with its SINR, we actually get an interval of the link rate. 

Suppose we use the lower bound of the link rate as the input of MTSP and get the objective value, i.e., schedule length, ${\hat T}$; then use the upper rates to get the schedule length ${\check T}$. We prove in the following that the true optimal schedule length ${T^*}$ is bounded by these two values, that is, ${\check T} \leq {T^*} \leq {\hat T}$.


\begin{theorem}
\label{theo:bounds}
The optimal scheduling length of MTSP that can be approximated by the MCCR model is bounded by ${\hat T}$ and ${\check T}$.
\end{theorem}

\begin{proof}
Consider the lower bound ${\check T}$, which is derived from the upper bound for link rate, i.e., $\hat r_{i\Cc}$. In the LP dual (\ref{eq:dual}), we change the coefficient $r_{i\Cc}$ in (\ref{eq:dualcons}) from the actual value of the rate to its upper bound $\hat r_{i\Cc}$. As $r_{i\Cc} \leq \hat r_{i\Cc}$, it can be easily verified that any solution of the problem with the new rate setting is also a feasible solution of the original problem. That is, the latter has a larger solution space and thus a greater or equal objective value, which equals to the optimal schedule length due to the strong duality of LP. Hence we conclude ${T^*} \geq {\check T}$. 

Following the similar approach, we get ${T^*} \leq {\hat T}$, so the conclusion follows.
\end{proof}

Note although the schedule solution derived from the lower bound may not be attainable for the actual network, it contributes to tighten the scope of optimal solution.  For the networks that cannot be exactly mapped to the MCCR model, as proved in previous work, the MTSP is generally hard. Therefore, upper and lower bounds are critical on defining the optimum region and evaluating any heuristic result. 

\subsection{Improvement on the Upper Bound}
\label{sec:vd}

The upper bound ${\hat T}$ can be further improved as the follows. Start with the optimal solution of ${\hat T}$ and its corresponding activated group set ${\hat\CH} = \{\Cc \in \CH: {\hat t}_\Cc > 0\}$, if we back to the original rate setting and activate all the groups in ${\hat\CH}$, according to the rate bound, $r_{i\Cc}$ is either increased or kept the same value. In the case that some of the queues become empty before $\hat t_{\Cc}$, we will continue with the new group $\Cc'$, which contains all none-zero queues in this group $\Cc$, until all the queues are empty. Clearly, $\Cc'\subset \Cc$, by the monotonicity property of rate, $r_{i\Cc} \leq r_{i\Cc'}$. Thus to empty the same demand as $r_{i\Cc}{\hat t}_{\Cc}$, less or at most the same time duration is needed. In this way, we construct a feasible solution of the MTSP with true link rates and observe its scheduling length ${\hat T'} \leq {\hat T}$. As $T^*$ is the optimal solution, implying $T^*\leq {\hat T'}$, we get $T^*\leq {\hat T'} \leq {\hat T}$, showing that ${\hat T'}$ is a tighter upper bound for $T^*$.

\section{Simulation Setup and Results}
\label{sec:simulation}

\subsection{Simulation Setup}

We consider wireless networks with $N$ transmitters randomly placed in a square area of  $1000 \times 1000$ meters with receivers concentrate at the centre. All the links are activated with a given power of $30$ dBm. The background noise is set to $-100$ dBm. The wireless signal propagation follows a distance-based model with a path loss exponent of 4, thus the distance between each pair of transmitter and receiver, i.e., the length of a link, is restricted between 3 and 250 meters to obtain links of practically meaningful signal-to-noise ratio (SNR) values. For the rate values, we utilize the Shannon function, assuming interference from concurrently activated links as noise.

\subsection{Simulation Results for Different Values of $K$}

The parameter $K$ stands for the total number of clusters that the links are divided into. For wireless networks with $N = 15$ links and the links have uniformly random queue sizes of [100, 1500] bits, we performed the simulations with $K= \{ 1, ~2, ~3, ~4, ~5,~6 \}$ respectively. For each setup, 100 instances are run and the results are normalized with respect to the global optimal solution of LP in (\ref{eq:lp}), which is obtained by AMPL \cite{AMPL02}. 

Figure \ref{fig:kk} presents the average upper and lower bounds. As expected, enlarging $K$ improves the performance of the result. Both upper and lower bounds become tighter with the increase of $K$. The trend is more noticeable for the first three values of $K$. For $K=4$ and beyond, the average gap between the lower and upper bounds is less than 9 percent and the improvement caused by $K$ is not pronounced as before. The observation implies that $K$ can be set to a relatively small number in relation to $N$, which means we can get satisfactory results with low computational complexity.

\begin{figure} [ht!]	
    \centering
{\includegraphics[width=120mm]{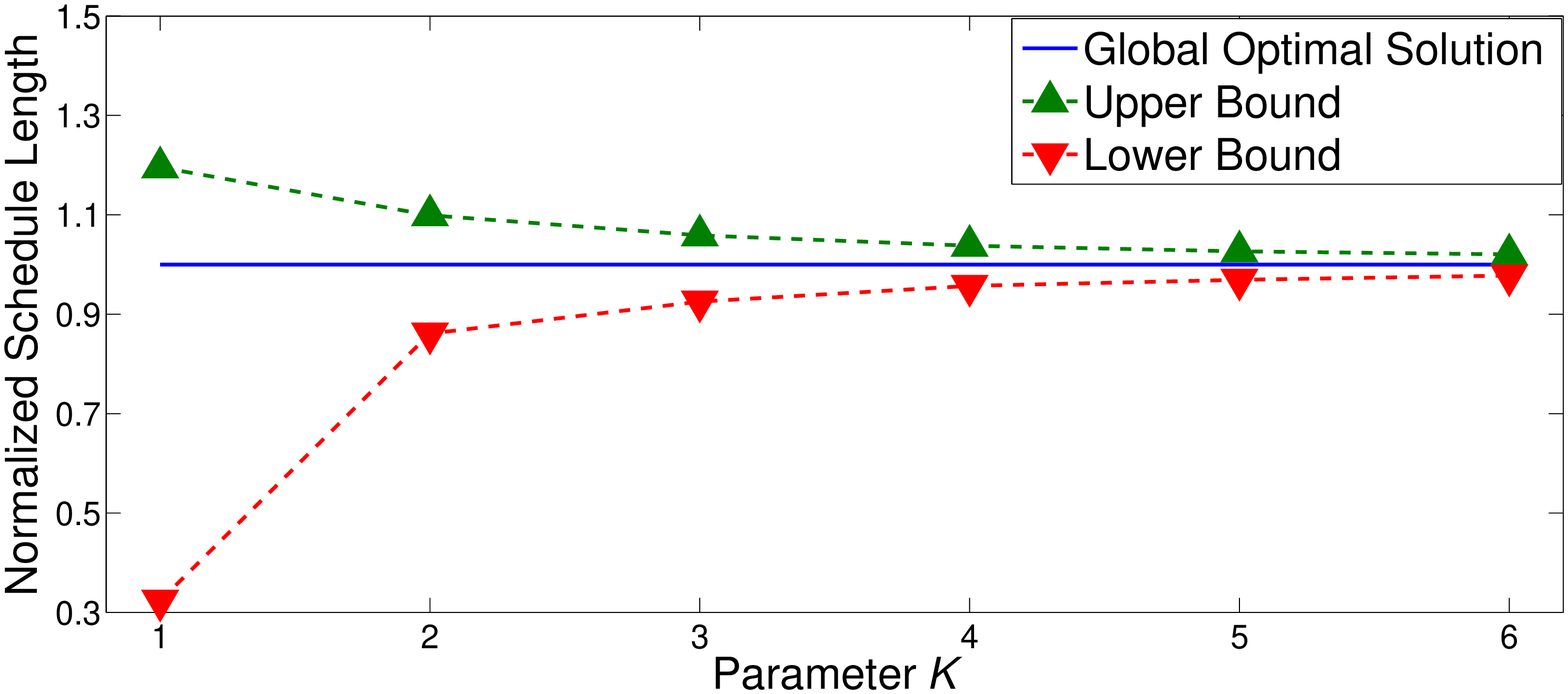}}
\vspace{-5mm}
\caption{Effect of $K$ on upper and lower bounds of schedule length.}
\label{fig:kk}
\end{figure}

We use the mean value of link length in each cluster, i.e., $\mu^*$, to calculate the channel gains and consequently, the link rates. This approach provides an approximate result of the optimal solution. Figure \ref{fig:cdfk} shows the empirical cumulative distribution function (CDF) of the normalized results with different values of $K$. For any group and activated link, the rate derived from $\mu^*$ can be either larger or smaller than its true value, thus the result of this approach is possibly less or greater than the true optimum (1.0) in schedule length. Overall, we can see the approach provides an excellent approximation of the global optimal solution. And, with the increase of $K$, the approximate results have less deviation from the optimal schedule length. For all the 100 instances, the optimality gap is no more than 3 percent when $K \geq 3$.

\begin{figure} [ht!]	
    \centering
{\includegraphics[width=120mm]{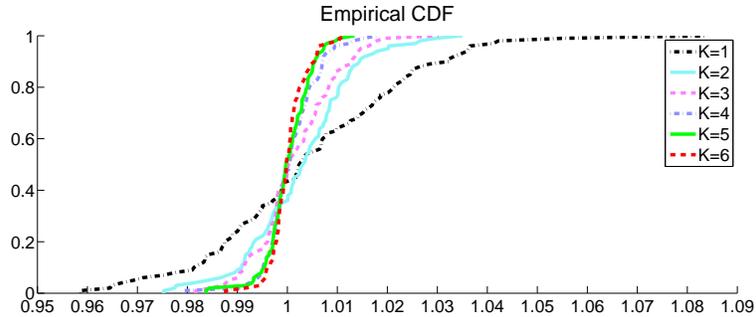}}
\vspace{-5mm}
\caption{Empirical CDF of normalized approximate schedule length.}
\label{fig:cdfk}
\end{figure}

\subsection{Simulation Results of Medium-Sized Networks}

We performed the proposed setup in medium-sized networks with $N = 30$ links, for which the true optimum can be hardly achieved by solving the LP in (\ref{eq:lp}) due to the exponentially increased complexity. The links have uniformly random queue sizes of [100, 3000] bits. 
We set $K=6$ and 100 instances are tested. The results are presented in Figure \ref{fig:30lk}. Herein, the lower bounds are normalized to 1.0 and the other values are set in relation to the lower bounds.

\begin{figure} [ht!]	
    \centering
{\includegraphics[width=120mm]{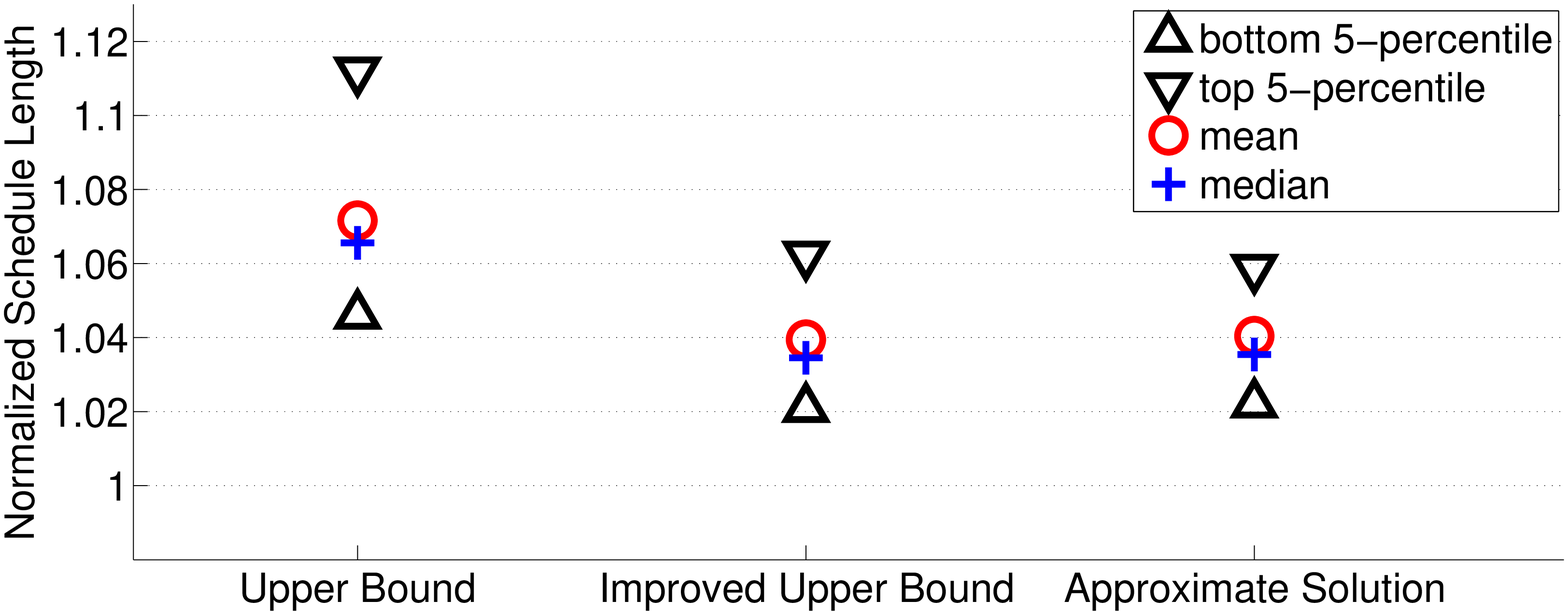}}
\vspace{-5mm}
\caption{Simulation on wireless networks with 30 links.}
\label{fig:30lk}
\end{figure}


The results show that the approach performs well on providing tight lower and upper bounds. The average gap between upper and lower bound is 7.2 percent. With using improved upper bounds in Section \ref{sec:vd}, the gap can be further decreased to 4.0 percent in average. The approximate solutions are very close to the relative lower bounds, with the gaps ranging from 2.2 to 5.9 percent for 90 instances of the 100.

\section{Conclusions}
\label{sec:conclusion}

We have provided fundamental insights on the minimum-time scheduling problem with multi-cluster cardinality-based rates, showing it can be solved in polynomial time and presenting conditions for problem decomposition. Then the model has been extended to more general wireless networks with k-means clustering and a column generation algorithm has been designed. We have applied the proposed approach to exactly solve the problem or provide a good approximation, as well as the lower and upper bounds for the true optimum. Numerical results have been provided to evaluate the performance of the approach. It has been demonstrated that, with a moderate number of cluster in the approximation, the optimality gap of no more than 4.0 percent on average can be reached. 

\appendices


\begin{thebibliography}{1}


\bibitem{j12}
V. Angelakis, A. Ephremides, Q. He, and D. Yuan.
\newblock Minimum-time link scheduling for emptying backlogged traffic in wireless systems: solution characterization and algorithmic framework.
\newblock Accepted by {\em IEEE Transactions on Information Theory}, DOI: 10.1109/TIT.2013.2292065.


\bibitem{Ar84}
E.~Arikan.
\newblock Some complexity results about packet radio networks.
\newblock {\em IEEE Transactions on Information Theory}, 30:910--918, 1984.

\bibitem{BjVaYu03}
P.~Bj{\"o}rklund, P. V{\"a}rbrand, and D. Yuan.
\newblock Resource optimization of spatial TDMA in ad hoc radio networks: a column generation approach.
\newblock In {\em Proc.\ of IEEE INFOCOM '03}, 2003.

\bibitem{BjVaYu04}
P.~Bj{\"o}rklund, P. V{\"a}rbrand, and D. Yuan.
\newblock A column generation method for spatial TDMA scheduling in ad hoc networks.
\newblock {\em Ad Hoc Networks}, 2:405--418, 2004.

\bibitem{BoEp06}
S.~A.~Borbash and A.~Ephremides.
\newblock Wireless link scheduling with power control and SINR constraints.
\newblock {\em IEEE Transactions on Information Theory}, 52:5106--5111, 2006.

\bibitem{BoLiXi10}
C~Boyac{\'y}, B.~Li, and Y.~Xia.
\newblock An investigation on the nature of wireless scheduling.
\newblock In {\em Proc.\ of IEEE INFOCOM '10}, 2010.

\bibitem{CaFiGuYu13}
A.~Capone et al.
\newblock Resource optimization in multi-radio multi-channel wireless mesh networks.
\newblock In: S.~Basagni, M.~Conti, S.~Giordano, and I.~Stojmenovic (editors),
{\em Mobile Ad Hoc Networking: The Cutting Directions}, 2nd edition,
Wiley and IEEE Press, 2013.

\bibitem{AMPL02}
R.~Fourer, D.~M.~Gay, and B.~W.~Kernighan.
\newblock  {\em AMPL: A Modeling Language for Mathematical Programming}. Duxbury Press, 2002.

\bibitem{GeNeTa06}
L.~Georgiadis, M.~J.~Neely, and L.~Tassiulas.
\newblock Resource allocation and cross-layer control in wireless networks.
\newblock {\em Foundations and Trends in Networking}, 1(1), 2006.

\bibitem{GoPsWa07}
O.~Goussevskaia, Y.~A.~Pswald, and R.~Wattenhofer.
\newblock Complexity in geometric SINR.
\newblock In {\em Proc.\ of ACM MobiHoc '07}, 2007.

\bibitem{GrHaNi00}
J.~Gr{\"o}nkvist, A.~Hansson, and J.~Nilsson.
\newblock A comparison of access methods for multihop ad hoc radio networks.
\newblock In {\em Proc.\ of IEEE VTC '00}, pp. 1435--1439, 2000.

\bibitem{HaSa88}
B.~Hajek and G.~Sasaki.
\newblock Link scheduling in polynomial time.
\newblock {\em IEEE Transaction on Information Theory}, 34:910--917, 1988.

\bibitem{isita2}
Q. He, V. Angelakis, A. Ephremides and D. Yuan.
\newblock Revisiting minimum-length scheduling in wireless networks: an algorithmic framework.
\newblock In {\em Proc.\ of IEEE ISITA '12}, 2012.

\bibitem{KoWi10}
S.~Kompella et al.
\newblock On optimal SINR-based scheduling in multihop wireless networks.
\newblock {\em IEEE/ACM Transactions on Networking}, 18:1713--1724, 2010.

\bibitem{Md03}
D.~J.~C.~MacKay.
\newblock An example inference task: clustering.
\newblock {\em Information Theory, Inference and Learning Algorithms}, ch. 20, Cambridge University Press, 2003.

\bibitem{PaEp08}
A.~Pantelidou and A.~Ephremides.
\newblock Minimum schedule lengths with rate control in wireless networks.
\newblock In {\em Proc.\ of IEEE MILCOM '08}, 2008.

\bibitem{Pr89}
C.~G.~Prohazka.
\newblock Decoupling link scheduling constraints in multihop packet radio
networks.
\newblock {\em IEEE Transactions on Computers}, 38:455--458, 1989.

\bibitem{StAm90}
D.~S.~Stevens and M.~H.~Ammar.
\newblock Evaluation of slot allocation strategies for TDMA protocols
in packet radio networks.
\newblock In {\em Proc.\ of IEEE MILCOM '90}, pp. 835--839, 1990.





















\end{thebibliography}
\end{document}